\documentclass[journal]{IEEEtran}
\usepackage{amsmath,mathtools,amsfonts,amssymb,amsthm}
\usepackage{float}
\usepackage{algpseudocode}
\usepackage[ruled,vlined]{algorithm2e}
\usepackage{array}
\usepackage[caption=false,font=normalsize,labelfont=sf,textfont=sf]{subfig}
\usepackage{textcomp}
\usepackage{stfloats}
\usepackage{url}

\usepackage{url}
\usepackage{xcolor}
\usepackage{verbatim}
\usepackage{graphicx}
\usepackage{soul}
\def\BibTeX{{\rm B\kern-.05em{\sc i\kern-.025em b}\kern-.08em
    T\kern-.1667em\lower.7ex\hbox{E}\kern-.125emX}}
\usepackage{balance}
\usepackage{multirow}
\usepackage{booktabs} % To thicken table lines

\usepackage{caption}                         % <--- added
 \usepackage{multirow}
 \usepackage{threeparttable}
\usepackage{mathtools}  % For \coloneqq and other math tools (optional if not using \coloneqq)
\PassOptionsToPackage{linesnumbered,ruled,vlined}{algorithm2e}
\usepackage{tablefootnote}
\usepackage[numbers,sort&compress]{natbib}

\newtheoremstyle{sltheorem}
{}                % Space above
{}                % Space below
{}        % Theorem body font % (default is "\upshape")
{10pt}                % Indent amount
{\bfseries}       % Theorem head font % (default is \mdseries)
{:}               % Punctuation after theorem head % default: no punctuation
{ }               % Space after theorem head
{}                % Theorem head spec
\theoremstyle{sltheorem}
\newtheorem{assumption}{Assumption}
\newtheorem{theorem}{Theorem}
\newtheorem{definition}{Definition}

\begin{document}
\raggedbottom

\title{Duration-Aware Ramp Adequacy Screening}

\author{Qian Zhang, Aidan Looney, Chao Tian, Xu Andy Sun,  Le Xie
\thanks{Qian Zhang, Aidan Looney, and Le Xie are with the School of Engineering and Applied Sciences, Harvard University, USA. Chao Tian is with the Department of Electrical and Computer Engineering, Texas A\&M University, USA. Xu Andy Sun is with the Sloan School of Management and the MIT Energy Initiative, Massachusetts Institute of Technology, USA. (correspondence e-mail: qianzhang@g.harvard.edu).}
}

\markboth{Submitted to the IEEE Transactions on Energy Markets, Policy and Regulation}%
{Zhang \MakeLowercase{\textit{et al.}}: Duration-Aware Ramp Adequacy Screening}

\maketitle

\begin{abstract}
Ramp products are widely used in regional electricity markets to procure intertemporal flexibility in anticipation of net demand changes. However, the design of such ramp products often lacks a clear specification of ramping duration, potentially leading to infeasible dispatch solutions. This paper develops a duration-aware ramp adequacy screening method that evaluates whether the currently committed and dispatched fleet can meet the anticipated net-demand ramp requirement across different durations. A negative ramp adequacy margin identifies an insufficient-duration set in which the fleet lacks adequate ramp capability. Evaluating this margin across duration supports product-duration selection, while tracking it over time provides a metric for assessing ramp adequacy under different products and dispatch policies. We further formulate rolling-window ramp-reserve procurement with horizon-dependent forecast uncertainty and show that a product can affect ramp capability beyond its designated duration through changes in dispatch positioning. Building on this framework, we develop a forecast-free ramp-reserve scarcity dispatch policy that prioritizes resources according to their remaining ramp-up durations and, in the transmission-unconstrained setting, achieves the same minimum operational security loss as a perfect-foresight benchmark. Studies on a 10-generator system and a 2751-bus synthetic Texas grid demonstrate the value of the proposed framework for early detection of ramp scarcity, product design and evaluation, and screening-guided emergency dispatch.
\end{abstract}


\section{Introduction}

Increasing renewable penetration, particularly solar generation, has made ramp capability a growing operational concern for many regional transmission organization. During sustained net-load transitions, the binding limitation may be the deliverable ramp of the committed fleet rather than its installed generation capacity. If the system cannot follow an upward change in net demand, the consequences may include scarcity prices, emergency actions, and ultimately load shedding \cite{thatte2014analysis,Ela2008}.

There are broadly two ways to improve ramp adequacy. The first is to expand physical system flexibility through investments in resources such as energy storage, fast-start generation, and demand response \cite{kargarian2016multi,shahmohammadi2018role}. The second is to make better use of the flexibility already available in the committed fleet through improved dispatch policies and market mechanisms \cite{wang2013two,gu2016stochastic}. The latter is particularly important in operations, where operators often need to manage an approaching ramp event primarily using resources that are already online. This paper focuses on the latter operational challenge.

To preserve flexibility for net-demand inter-temporal variations during real-time operation, look-ahead dispatch (LAD) and ramp reserve products are two widely used approaches. Multi-interval LAD explicitly accounts for future \emph{expected} ramping needs, while ramp products reserve upward and downward capability to address both expected net-demand \emph{variability} and forecast \emph{uncertainty}. For example, Midcontinent Independent System Operator (MISO) procures a 10-minute ramp capability product that covers both expected ramping needs and forecast uncertainty \cite{wang2016ramp,miso2022ramp}, while California ISO (CAISO) operates a flexible ramping product in conjunction with LAD \cite{caiso2016flexible}. Prior research has explored improved ramp-product formulations, multi-duration product designs, uncertainty-aware requirements, and related market refinements \cite{wang2016enhancing,wang2014flexible}. Forecast uncertainty has also motivated stochastic, robust, and chance-constrained approaches to dispatch and reserve procurement \cite{lorca2014adaptive,gu2016stochastic,zhang2024efficient}.

Despite these developments, two practical questions remain. First, \emph{which ramp-product durations are needed given the current commitment and dispatch state, together with the anticipated net-demand trajectory?} A fixed 10-, 30-, or 60-minute product describes a procurement horizon, but it does not reveal whether the underlying shortage is primarily a short-duration ramp-rate problem or a longer-duration capacity headroom problem. Moreover, the capability available at a given duration depends on the current dispatch state. A product that is adequate when cleared may become less effective as subsequent economic dispatch moves generators toward their capacity limits \cite{qian2025}.

Second, \emph{what should the operator do when ramp-product procurement does not provide sufficient margin?} Increasing reserve procurement can preserve additional flexibility, but its effectiveness is inherently limited. During severe ramp reserve scarcity, the general idea of switching operating objectives under stressed conditions dates back to the preventive, emergency, and restorative operating states introduced in \cite{liacco1967adaptive}. While emergency controls have been studied for many power-system security problems, relatively little attention has been given to dispatch policies specifically designed to preserve ramp capability during an emerging ramp shortage. 

This paper addresses both questions through a \emph{duration-aware ramp adequacy screening} framework. We define the remaining ramp-up duration of a resource as the length of time for which it can continue ramping at its maximum rate before reaching capacity. Aggregating these resource states gives the fleet capability curve $D_t(k)$ as a function of ramp duration $k$. Comparing this curve with the forward ramp requirement $R_t(k)$ yields the ramp adequacy margin:
\begin{equation*}
M_t(k)=D_t(k)-R_t(k).
\end{equation*}
A negative margin identifies durations at which the committed fleet lacks sufficient endpoint ramp capability. Evaluating $M_t(k)$ across duration identifies the insufficient-duration set and supports ramp-product duration selection, while tracking the same metric over time evaluates how ramp adequacy evolves with dispatch, procurement, and deployment.

The remaining-duration perspective also motivates a new dispatch strategy, \emph{ramp-reserve scarcity} (RS) dispatch, which preserves future ramp capability by prioritizing resources according to their remaining ramp-up durations. We show that, without transmission constraints, RS dispatch achieves the same minimum operational security loss as a perfect-foresight benchmark, establishing a fundamental security limit for ramp-product designs. We further extend this principle to transmission-constrained systems through an implementation compatible with conventional security constrained economic dispatch (SCED).

The main contributions of this paper are threefold. First, we develop a duration-aware ramp adequacy screening framework based on $M_t(k)$ that identifies insufficient-duration sets by comparing fleet ramp capability with the anticipated ramp requirement across duration. Second, we use the remaining-duration representation to explain how ramp-product duration, forecast uncertainty, and subsequent dispatch jointly determine post-clearing ramp capability. Third, we establish the forecast-free security limit of RS dispatch and develop a network-constrained implementation that can be activated as an emergency operating policy. Case studies on a 10-generator system and a 2751-bus synthetic Texas grid demonstrate these findings.

\section{Preliminaries} \label{sec:pr}

\subsection{System Description}

We denote by $\mathcal{N}  \doteq \{G_1, G_2, ..., G_N\}$ the set of dispatchable resources available in the system. Let $\overline{g}_i$ and $\overline{r}_i$ denote the maximum power capacity and ramp-up capability of resource $G_i$, respectively. For simplicity, all resources are assumed to have zero minimum power output, i.e. $\underline{g}_i = 0$, and symmetric ramping capabilities, i.e. $\underline{r}_i = \overline{r}_i$. 
This paper mainly focuses on system operation after the unit commitment decisions have been determined, over the time horizon $\mathcal{T}:= \{1,2,..,T\}$.

Given the \emph{net} load profile $\{d(t):t \in \mathcal{T} \}$, the decision maker determines an appropriate dispatch policy $\mathcal{G}$ that accounts for both economic efficiency and operational security. Let $g_i(t)$ be the power output of resource $i$ at time $t$. Two types of hard constraints are considered in this paper: power capacity constraints (\ref{pc}) and ramp capability constraints (\ref{rc}).
\begin{subequations} \label{hardconstraints}
\begin{align}
0 \leq g_i(t) \leq \overline{g}_i,  \label{pc} \\
-\overline{r}_i \leq g_i(t) - g_i(t-1) \leq \overline{r}_i.  \label{rc}
\end{align}
\end{subequations}

Consistent with standard market-clearing models, the power-balance constraints are relaxed to \emph{soft} constraints at the dispatch timescale, where at each dispatch interval an energy shortfall is represented by a non-negative load-shedding slack variable, denoted $s(t)$. The power-balance constraint is relaxed to (\ref{softpower}):

\begin{equation} \label{softpower}
d(t) - s(t)  = \sum_i^N g_i(t).
\end{equation}
Let $\rho_s$ denote the penalty associated with load shedding $s(t)$. Economically, $\rho_s$ plays the role of the scarcity price, set to the cost of emergency action or to the value of lost load (VoLL). In practice, the ISO first dispatches expensive offline fast-start resources before curtailing firm load. For example, MISO sets the administrative penalty for fast-start energy to about \$2{,}500/MWh \cite{miso2022ramp,wang2016ramp}, roughly two orders of magnitude above the cost of conventional resources, while firm curtailment is valued at a VoLL of \$10{,}000/MWh \cite{miso2024scarcity}.

\emph{Remark:} In the security-constrained implementation, network flow limits are represented using the standard power-transfer distribution factor (PTDF) formulation in the synthetic Texas case study (Section \ref{sec:texas}). Most of the analytical development that follows omits network constraints to better illustrate the value of ramping reserve procurement.

\subsection{Operational Security Loss}

\begin{definition} \label{reliabityloss}
[Operational Security Loss] Given the realized net load profile ${d(t)}$ and dispatch interval $\Delta t$, the operational security loss associated with a dispatch policy $\mathcal{G}$ is defined as the total amount of demand shedding over the time horizon $\mathcal{T}$, i.e.,
\begin{equation} \label{totalloss}
L(\mathcal{G},{d(t)}) := \sum_{t \in \mathcal{T}}  \max\{0, d(t)-\sum_i^N g_i(t) \} \Delta t.
\end{equation}
\end{definition}

The analysis primarily focuses on operational security losses caused by ramping reserve scarcity. As the penetration of renewable resources increases, ramping reserve needs generally arise from two sources: net load \emph{variation} and net load forecast \emph{uncertainty}. Depending on the market design, different ISOs may define ramping reserve products to address only forecast uncertainty or to accommodate both sources \cite{qian2025}.

\begin{assumption}[Forecasting Error]
\label{error}
Let $\hat{d}_t(t+\tau)$ be the forecast made at time $t$ for the net demand at time $t+\tau$. The current-interval value is assumed to be known:
\begin{equation} \label{firsterror}
\hat{d}_t(t)=d(t).
\end{equation}
For $\tau>0$, the forecast error satisfies
\begin{equation}
\hat{d}_t(t+\tau)-d(t+\tau)\sim\mathbb{P}_\tau,
\end{equation}
where $\mathbb{P}_\tau$ depends on the forecasting method, forecast horizon, and information available at time $t$. Because forecast uncertainty generally increases with the look-ahead horizon, ramping reserve procurement is typically implemented on a rolling-window basis.
\end{assumption}

\subsection{Rolling-Window Ramp-Reserve Procurement} \label{sec:rolling}

Ramp products provide a procurement-based mechanism for preserving capability to meet selected future ramping requirements. Rather than scheduling the full future energy trajectory, a duration-$W$ product reserves the capability needed to cover the forecast net-demand change over that horizon \cite{qian2025}. In practice, procurement is implemented on a rolling-window basis. As illustrated in Fig.~\ref{fig:rolling}, the system re-evaluates and procures 5- and 10-minute ramp products at every 5-minute dispatch interval.

\begin{figure}[H]
\centering
\includegraphics[width=\columnwidth]{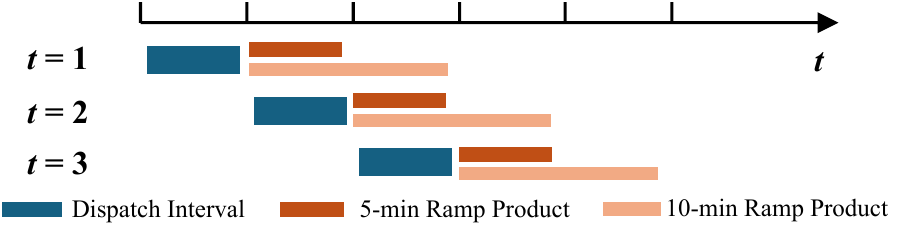}
\caption{Rolling-window ramp-reserve procurement.}
\label{fig:rolling}
\end{figure}

Let $\mathcal{W}=\{0,\ldots,W\}$. A duration-$W$ ramp product procures upward and downward capability that can be delivered from the current dispatch state over $W$ intervals. The future net-demand uncertainty is represented by
$\left[ \hat d_t(t+W)-\delta(W),\hat d_t(t+W)+\delta(W)\right].$
The associated dispatch policy is denoted by
$\mathcal{G}^{\mathrm{ED}_t^{W}}$, with unit-specific constraints applied to
$i\in\mathcal N$.

\begingroup\footnotesize
\begin{equation}
\label{eq:product}
\begin{aligned}
\min \quad &
\sum_{i=1}^N c_i g_i(t)
+\rho_s\big(s(t) + s(t+W)\big)
\\
\text{s.t.}\quad
&
-\underline r_i
\leq g_i(t)-g_i(t-1)
\leq \overline r_i,\quad 0 \leq g_i(t)
\leq \overline g_i,
\\
&
0\leq r_i^+(t+W)\leq W\overline r_i,
\quad
g_i(t)+r_i^+(t+W)\leq\overline g_i,
\\
&
0\leq r_i^-(t+W)\leq W\overline r_i,
\quad
g_i(t)-r_i^-(t+W)\geq0,
\\
&
\sum_{i=1}^N r_i^+(t+W)
\geq
\hat d_t(t+W)+\delta(W)-s(t+W)
-\sum_{i=1}^N g_i(t),
\\
&
\sum_{i=1}^N r_i^-(t+W)
\geq
\sum_{i=1}^N g_i(t)
-\hat d_t(t+W)+\delta(W)+s(t+W),
\\
&
\sum_{i=1}^N g_i(t)+s(t)=d(t), \quad s(t) \geq 0, \ s(t+W) \geq0.
\end{aligned}
\end{equation}
\endgroup

The upward and downward requirements are measured directly from the current served load to the upper and lower bounds of the forecast uncertainty interval at $t+W$ after subtracting the shedding quantity $s(t+W)$, penalized by $\rho_s$. The decision variables $r_i^+(t+W)$ and $r_i^-(t+W)$ therefore represent the upward and downward capability reserved from the current dispatch state and
deliverable over duration $W$. 

In practice, insufficient downward ramp capability is generally observed less frequently and is typically less consequential than upward scarcity, as excess downward imbalance can often be managed through actions such as renewable curtailment. We therefore do not focus on downward scarcity in this study. At each subsequent clearing interval, the forecast, uncertainty band, ramp-product requirements, and current dispatch state are updated. For combined-duration configurations, \eqref{eq:product} is imposed at each selected duration.

\section{Duration-Aware Ramp Adequacy Screening} \label{sec:screening}

Duration-aware ramp adequacy screening evaluates whether the current fleet can support an anticipated net-demand ramp across different durations. We compare fleet ramp capability with anticipated net-load changes and inform ramp-product selection and operational monitoring.

\subsection{Remaining Ramp-Up Duration}

Ramp rate and capacity headroom both determine sustained ramp deliverability. Their ratio determines how long a resource can continue increasing output before reaching its upper limit, which motivates the definition of \emph{remaining ramp-up duration}

\begin{definition}[Remaining Ramp-Up Duration] \label{defin2}
The remaining ramp-up duration $l_i(t)$ is the time for which a resource $i$ can continue at its maximum upward ramp rate before reaching its capacity at time $t$:
\begin{equation} \label{remain}
l_i(t) = \frac{\overline{g}_i-g_i(t)}{\overline{r}_i}.
\end{equation}
\end{definition}

The remaining ramp-up duration $l_i(t)$ captures this effect by measuring how long a committed generator can continue ramping upward from its current output before reaching its capacity limit. Two resources with the same ramp rate may have very different remaining durations if they occupy different dispatch positions. To characterize how this state evolves with dispatch, let $u_i(t)\in[-1,1]$ denote the normalized ramp instruction during interval $t$, such that
$ g_i(t)=g_i(t-1)+u_i(t)\overline{r}_i$.
Substituting this expression into Definition~\ref{defin2} gives
\begin{equation} \label{eq:lrecur}
l_i(t)=l_i(t-1)-u_i(t).
\end{equation}
Thus, an upward dispatch instruction reduces the remaining duration, whereas a downward instruction restores it. Under the capacity constraints, $0\leq l_i(t)\leq \overline{l}_i$, where $\overline{l}_i=\overline{g}_i/\overline{r}_i$. Consequently, the vector $\{l_i(t)\}_{i=1}^N$ records how the realized dispatch trajectory has positioned the fleet for a future ramp event. For energy-limited resources, the remaining ramp-up duration is additionally bounded by the available energy headroom.

\subsection{Fleet Capability and Ramp Requirement Curves}

Consider a ramp event beginning at time $t$. Suppose that the system is initially balanced, $\sum_{i=1}^N g_i(t)=d(t)$, and focusing on ramp-up period when net demand is nondecreasing over $k= \{ 1,\dots,K\}$. Aggregating the resource-level states gives the \emph{fleet ramp capability curve}:
\begin{equation} \label{eq:capcurve}
D_t(k) \coloneqq \sum_{i=1}^N \overline{r}_i \min \{l_i(t),k\}.
\end{equation}

The forward change in net demand defines the corresponding \emph{ramp requirement curve}:
\begin{equation} \label{eq:reqcurve}
R_t(k) \coloneqq d(t+k)-d(t).
\end{equation}

The contribution of resource $i$ to $D_t(k)$ is limited by two quantities: it can increase at no more than $\overline{r}$ per interval, and it can sustain that increase for no more than $l_i(t)$ remaining duration. Thus, $D_t(k)$ is the maximum increase in fleet output, relative to time $t$, that can be delivered by time $t+k$. The curve is nondecreasing and discretely concave in duration: its short-duration slope is the aggregate ramp rate of resources that still have headroom, while its long-duration limit is the fleet's available capacity headroom. 

Equation~\eqref{eq:reqcurve} uses the realized net-demand path for retrospective evaluation. In real-time screening, the operator can replace $d(t+k)$ with the forecast-adjusted upper trajectory $\hat d_t(t+k)+\delta(k)$, consistent with the uncertainty treatment in \eqref{eq:product}. This distinction changes the information supplied to the requirement curve but not the construction of the fleet capability curve.

\subsection{Ramp Adequacy Margin}

The difference between the capability and requirement curves defines the paper's main screening metric. The \emph{ramp adequacy margin} is
\begin{equation} \label{eq:margin}
M_t(k)\coloneqq D_t(k)-R_t(k).
\end{equation}
A negative margin indicates that the ramp requirement over duration $k$ exceeds the maximum increase deliverable from the current fleet state. The corresponding shortfall is $-M_t(k)$. A nonnegative margin indicates sufficient endpoint capability over duration $k$, but does not by itself guarantee that the full net-demand trajectory is feasible under exact power balance.

The set
\begin{equation} \label{eq:set}
\mathcal{B}_t\coloneqq \{k:M_t(k)<0\}
\end{equation}
is the \emph{insufficient-duration set} associated with the system state at time $t$. Its location identifies the exposed durations, while the depth of the negative margin quantifies the corresponding ramp capability shortfall.

Because the resource states evolve according to \eqref{eq:lrecur}, the capability curve and the ramp adequacy margin must be recomputed as dispatch changes. Thus, $M_t(k)$ varies along two dimensions. Variation across $k$ identifies the durations that are exposed at a given time, while variation across $t$ shows when a selected duration becomes exposed or recovers. This state dependence distinguishes the proposed screen from static measures based only on installed ramp rates or procured reserve awards.

The ramp adequacy margin is therefore a screening metric and does not necessarily ensure complete feasibility. Actual trajectory feasibility also depends on how ramping is allocated among resources over time, since dispatching a resource less in an earlier interval may preserve its ramping capability for a later interval. This distinction motivates the duration-preserving dispatch policy developed in Section~\ref{sec:positioning}.

\subsection{Two-Axis Screening for Operations and Planning}
\label{sec:implications}

The ramp adequacy margin can be evaluated along two complementary
dimensions. Screening across duration at a fixed time supports
operational decisions, while screening across time at a fixed
duration supports planning and comparisons of dispatch policies
or ramp products.

\subsubsection{Fixed-Time Duration Screening for System Operations}

Fixed-time duration screening is intended primarily for day-ahead
and short-term operations. Once unit commitment and the initial
dispatch are established, the operator fixes the screening time
$t_0$ and evaluates $M_{t_0}(k)$ across ramp durations. Equivalently,
the fleet capability curve $D_{t_0}(k)$ is compared with the forward
ramp requirement curve $R_{t_0}(k)$.

The resulting insufficient-duration set $\mathcal{B}_{t_0}$
identifies durations that cannot be supported by the current fleet
state, while the magnitude of the negative margin quantifies the
corresponding capability shortfall. These results identify which
durations require additional protection and inform ramp reserve
procurement, redispatch, or activation of other flexible resources.
The screen can be updated when the dispatch state or net-demand
forecast changes materially during operation.

\subsubsection{Fixed-Duration Temporal Screening for Planning and Evaluation}

Fixed-duration temporal screening supports planning and
retrospective evaluation. For a selected duration $\bar{k}$,
the margin $M_t(\bar{k})$ is evaluated over time using historical
operating data or simulated system trajectories. Repeating this
analysis for several durations reveals how short-, medium-, and
long-duration ramp capability changes with system conditions.

This temporal view provides a common metric for comparing dispatch
policies, ramp reserve products, and procurement strategies.
For example, planners can examine how frequently a duration becomes
insufficient, how long the deficiency persists, and how the margin
changes after introducing a different product or dispatch policy.
These comparisons identify which approaches improve ramp capability,
at which durations, and under which operating conditions, informing
product selection and design before deployment in system operations.

\section{Ramp Products: Operating Mechanism and Fundamental Limits} \label{sec:limits}

Section~\ref{sec:screening} identifies an exposed duration band from the current dispatch state, but does not determine the best achievable security outcome. This section defines that outcome with a perfect-foresight benchmark, then identifies the dispatch that preserves ramp capability, and finally establishes a deterministic security limit with an implementable dispatch formulation.

\subsection{Perfect-Foresight Security Benchmark}

The perfect-foresight oracle defines the lowest operational security loss attainable over the full horizon:

\begin{equation} \label{oracle}
  \begin{aligned}
  \mathcal{G}^\text{oracle}: \min_{g_i(t),s(t)} & \quad \sum_{t \in \mathcal{T}} s(t)  \\
  \text { s.t.}  \quad &\text{for all } i \in \mathcal{N}  \text{ and } t \in \mathcal{T}:\\
  (\mu_{it}^-, \mu_{it}^+): \quad & 0 \leq g_i(t) \leq \overline{g}_i, \\
  (\rho_{it}^-, \rho_{it}^+): \quad & -\overline{r}_i \leq g_i(t) - g_i(t-1) \leq \overline{r}_i, \\
  \lambda_t: \quad & d(t) - s(t) = \sum_{i=1}^N g_i(t), \\
  \gamma_t: \quad & s(t) \geq 0,
  \end{aligned}
\end{equation}
where the multipliers $\{\mu_{it}^\pm,\rho_{it}^\pm,\lambda_t,\gamma_t\}$ correspond to the displayed constraints.

Because the oracle optimizes dispatch using the realized net-demand trajectory over the entire horizon, its demand shedding is no greater than that of \emph{any} other feasible dispatch policy. It therefore provides a lower bound on operational security loss. However, such a policy is not implementable under Assumption~\ref{error}, because the decision-maker has only imperfect information about future net demand, with forecast uncertainty generally increasing with the prediction horizon. The oracle therefore serves as a security benchmark against which implementable dispatch policies and procurement designs can be evaluated. The remaining question is whether this benchmark can be attained, or approximated, \emph{without} perfect foresight.

\subsection{How Ramp Products Work: Preserving Remaining Ramp-Up Duration}
\label{sec:positioning}

Ramp products do not create additional physical ramp capability. Instead, they
improve future ramp deliverability by changing the current dispatch position of
the fleet. This mechanism can be seen directly through the remaining ramp-up
duration $l_i(t)$. For a duration-$W$ upward ramp product, the reserve that
resource $i$ can provide is limited by both its ramp rate and its available
capacity headroom:
\begin{equation}
r_i^+(t+W)
\leq
\min\left\{W\overline r_i,\,
\overline g_i-g_i(t)\right\}
=
\overline r_i\min\left\{W,l_i(t)\right\}.
\label{eq:reserve_duration}
\end{equation}
Thus, the maximum $W$-duration upward capability can be reserved from a resource is determined directly by its remaining ramp-up duration. Aggregating \eqref{eq:reserve_duration} across all resources gives
\begin{equation}
\sum_{i=1}^N r_i^+(t+W)
\leq
\sum_{i=1}^N
\overline r_i\min\{W,l_i(t)\}
=
D_t(W).
\label{eq:product_capability_link}
\end{equation}
Therefore, satisfying a duration-$W$ ramp requirement requires the current
dispatch state to preserve sufficient fleet capability $D_t(W)$ at that
duration. If the economic dispatch would otherwise leave too little capability,
the ramp-product constraint changes the energy dispatch so that additional
headroom, and hence additional remaining ramp-up duration, is retained on
selected resources. From this perspective, the primary effect of a ramp product
is \emph{dispatch deliverability}: it limits how much of the fleet's future
ramping capability can be consumed by the current energy dispatch.

The evolution in Equation (\ref{eq:lrecur}) makes this mechanism particularly transparent. Dispatching a resource upward
reduces its remaining ramp-up duration, while dispatching it downward restores that duration. Conventional economic dispatch may repeatedly increase the output of cheaper resources and drive their $l_i(t)$ toward zero. While current demand remains balanced, the fleet becomes progressively less capable of sustaining a future ramp. A ramp product counteracts this effect by requiring the dispatch to retain enough remaining duration to satisfy the selected look-ahead requirement.

\begin{figure}[H]
  \centering
    \includegraphics[width=\columnwidth]{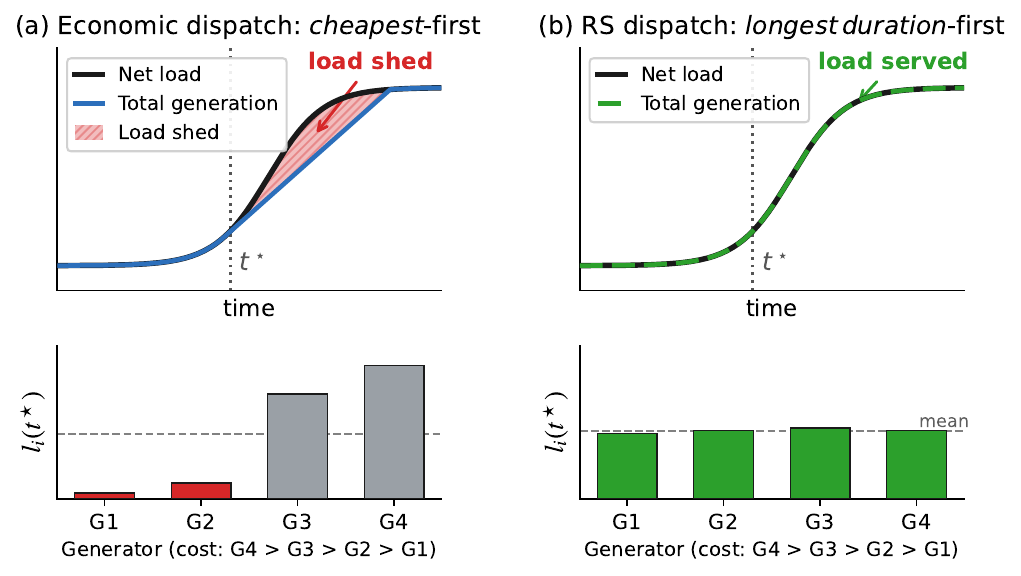}
    \caption{Effect of dispatch policy on the distribution and preservation of remaining ramp-up duration across the fleet.}
    \label{illu}
\end{figure}

Figure~\ref{illu} illustrates this distinction. In panel (a), economic dispatch
concentrates generation on lower-cost resources, leaving some of them close to
their capacity limits and with short remaining ramp-up durations. The same
system-wide generation can be supplied from a different dispatch position, as
illustrated in panel (b), in which remaining ramp-up duration is distributed
more evenly across the fleet. Such repositioning can also occur through
cross-ramping: some resources move downward while others move upward by the
same aggregate amount, leaving total generation unchanged while redistributing
the vector $\{l_i(t)\}_{i=1}^N$.

This interpretation also clarifies why product duration matters. A duration-$W$ product explicitly preserves capability at the selected duration $W$ through its constraint on $D_t(W)$, while its effect on other durations is only indirect. As dispatch evolves, the remaining ramp-up durations also change, so the preserved capability can subsequently be consumed. Ramp-product deliverability is therefore inherently state dependent. This perspective motivates a stronger dispatch principle: rather than preserving capability only at selected product durations, can the fleet be dispatched to preserve remaining ramp-up duration across \emph{all} durations? The next subsection develops such a policy and establishes its fundamental security limit.

\subsection{Deterministic Security Limit and Practical Implementation}
\label{sec:backstop}

The remaining-duration perspective naturally leads to a dispatch policy that preserves ramp capability as effectively as possible over time. When the system requires upward movement, resources with longer remaining ramp-up durations are dispatched first. When downward movement is required, resources with shorter remaining durations are prioritized.  This new dispatch policy is named as \emph{ramp-reserve scarcity} (RS) dispatch in this paper, which is inspired by optimal energy harvesting systems in communication networks \cite{tutuncuoglu2012optimum} and described in Algorithm~\ref{greedyalg} in the Appendix. 

In contrast to ramp products, which preserve capability at selected durations through forward-looking procurement requirements, RS dispatch relies \emph{only} on the current system ramping need and the current remaining-duration state. Theorem \ref{optimal} shows that this policy achieves the minimum possible operational security loss.

\begin{theorem}[Optimality of Ramp-Reserve Scarcity Dispatch]
\label{optimal}
Under Assumptions~\ref{error} and \ref{nocurt}, the operational security loss
of the greedy RS dispatch policy $\mathcal{G}^{\mathrm R}$ in
Algorithm~\ref{greedyalg} is no greater than that of any other feasible
dispatch policy. Hence, $\mathcal{G}^{\mathrm R}$ achieves the same security level as the perfect-foresight oracle $\mathcal{G}^{\mathrm{oracle}}$ defined in \eqref{oracle}.
\end{theorem}

\begin{proof}
See Appendix~\ref{app:proof}.
\end{proof}

Theorem~\ref{optimal} establishes RS dispatch as a deterministic security benchmark. Despite using only current-interval information, greedy RS dispatch achieves the same minimum security loss as the perfect-foresight oracle. Therefore, under the stated assumptions, no ramp-product design can provide a lower security loss over the same event. Ramp products may still offer a more economical solution by selectively preserving capability at anticipated durations, while RS dispatch represents the fundamental security limit.

The greedy RS policy establishes the deterministic security limit in the transmission-unconstrained setting of Theorem~\ref{optimal}. When transmission capacity constraints are considered, however, generators can no longer be dispatched solely according to their remaining ramp-up durations. Their locations in the network also determine how redispatch affects transmission flows. Therefore, the greedy priority rule may not be directly feasible under network congestion. Instead, the RS principle can be incorporated into a conventional network-constrained SCED by adding a penalty on the dispersion of remaining ramp-up durations to the objective function.

To incorporate transmission constraints using the DC power flow model, let $\mathcal{V}$ denote the set of buses, $\mathcal{G}(i)$ the set of generators located at bus $i\in\mathcal{V}$, $H\in\mathbb{R}^{L\times|\mathcal{V}|}$ the power transfer distribution factor (PTDF) matrix, and $\overline{f}\in\mathbb{R}^{L}_{+}$ the vector of transmission line limits. Consider nodal demand $d_i(t)$, load shedding $s_i(t)\geq0$, the net real-power injection $p_i(t)$, and let $p(t)=[p_1(t),\ldots,p_{|\mathcal{V}|}(t)]^\top$. The resulting RS-SCED formulation $\mathcal{G}^{\mathrm{R}}_{{SCED}_t}$ is
\begin{equation}
\begin{aligned}
\min_{\substack{\{g_j(t)\},\\ \{s_i(t)\},\,\eta}} 
 \;&
\sum_{j=1}^{N} f_j\big(g_j(t)\big)
+\rho_s\sum_{i\in\mathcal{V}}s_i(t)
+\beta\sum_{j=1}^{N}\big(l_j(t)-\eta\big)^2
\\
\text{s.t.}\quad
&
\underline g_j
\leq g_j(t)
\leq \overline g_j,
\quad \forall j,
\\
&
-\overline r_j
\leq g_j(t)-g_j(t-1)
\leq \overline r_j,
\quad \forall j,
\\
&
p_i(t) = \sum_{j\in\mathcal{G}(i)}g_j(t)
-d_i(t)+s_i(t),
\quad \forall i\in\mathcal{V},
\\
&
\sum_{i\in\mathcal{V}}p_i(t)=0, \quad s_i(t)\geq0,
\quad \forall i\in\mathcal{V},
\\
&
-\overline f
\leq Hp(t)
\leq \overline f.
\end{aligned}
\label{eq:sced_ramp_equalize_qp}
\end{equation}

Different from conventional SCED objective function, the third term introduces the RS dispatch principle. For a given set of remaining ramp-up durations, the optimal center is $\eta^\star = \frac{1}{N}\sum_{j=1}^{N}l_j(t)$, so that $\sum_{j=1}^{N}(l_j(t)-\eta)^2$ penalizes the dispersion of remaining ramp-up durations around their fleet-wide mean. A larger value of $\beta$ therefore places greater emphasis on avoiding excessive depletion of the remaining ramp-up duration of individual resources.

Unlike the greedy RS policy, the formulation in \eqref{eq:sced_ramp_equalize_qp} does not impose a fixed dispatch priority. Instead, the RS penalty is optimized jointly with generation cost subject to generator limits, ramp-rate limits, nodal power balance, and transmission constraints. Thus, when congestion prevents a resource with a longer remaining ramp-up duration from following the greedy priority, SCED selects an alternative network-feasible redispatch while still favoring preservation of the overall remaining-duration profile. Setting $\beta=0$ recovers conventional SCED, whereas $\beta>0$ introduces the duration-preserving RS principle while retaining a convex quadratic formulation.

Similarly, the same dispersion penalty, $\beta\sum_{j=1}^{N}\big(l_j(t)-\eta\big)^2$,
can be added to the LAED objective while retaining its look-ahead, ramping, and transmission constraints. This provides a network-constrained implementation of the RS principle for both real-time SCED and look-ahead dispatch.

\section{Case Studies}
\label{sec:case}

The case studies follow the operational sequence established by the proposed framework. First, the ramp adequacy margin is used to diagnose the insufficient-duration set and provide advance warning of ramp scarcity. Second, the screening results are used to evaluate candidate product durations, duration portfolios, and post-procurement deliverability. Finally, when procurement does not fully address the diagnosed exposure, the temporal margin is used to evaluate the timing and effectiveness of emergency RS
dispatch.

The 10-generator system provides a controlled setting for examining both screening axes under perfect forecasts and forecast uncertainty. The 2751-bus synthetic Texas system then demonstrates large-scale screening, product evaluation, and emergency RS implementation under transmission constraints.

\subsection{Ramp Adequacy Screening and Early Warning (10-Generator System)}

\emph{Setup.}
The 10-generator system contains 2561~MW of installed capacity and 120~MW of aggregate upward ramp capability per 5-minute interval. We consider a 4-hour sustained ramp in which net demand increases from  1777~MW to 2373~MW. This subsection assumes perfect forecasts to isolate the screening results.

\emph{Early Diagnosis of Ramp Scarcity.}
At the beginning of the ramp window, the fixed-time duration screen is nonnegative over all examined durations, indicating that the fleet can support the projected ramp from its initial dispatch state. As economic dispatch progresses, however, lower-cost resources move toward their capacity limits and lose remaining ramp-up duration. By interval 17, or 85~minutes into the event, the updated screen identifies a 60--85-minute insufficient-duration set. Baseline load shedding does not begin until interval 26, so the screen provides a 45-minute advance indication of the emerging ramp shortage. Figure~\ref{fig:duration}(a) shows the capability and requirement curves at the diagnosis time, while Figs.~\ref{fig:duration}(b) and (c) track the 10- and 60-minute temporal margins.

\emph{Implications for Product-duration Selection.}
The diagnosed set also explains the performance of different ramp-product durations. Because the exposure begins near 60~minutes, the 5- and 10-minute products do not materially change the outcome, and each case retains 9.86~MWh of shedding. The 60-minute product is aligned with the lower edge of the exposed set and reduces shedding to 3.86~MWh. RS dispatch eliminates shedding to numerical tolerance and serves as the theoretical security benchmark established by Theorem~\ref{optimal}.

A system that appears adequate at the beginning of a ramp can become duration-deficient as economic dispatch consumes remaining ramp capability. Ramp adequacy should therefore be re-screened as the dispatch state evolves.

\begin{figure*}[t]
\centering
\includegraphics[width=\textwidth]{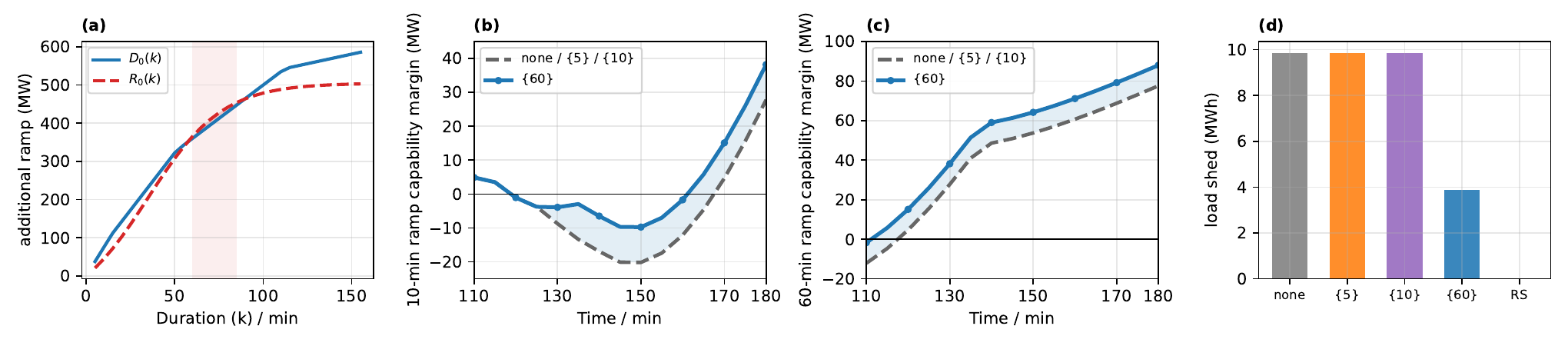}
\caption{Ramp adequacy screening and product comparison in the
10-generator system under perfect forecasts. (a) Fixed-time screen at time 0, (b) The 10-minute and (c) 60-minute temporal ramp capability reserve margins under different ramp product durations. (d) Load
shedding under different ramp product durations and RS dispatch.}
\label{fig:duration}
\end{figure*}

\subsection{Screening-Guided Product and Dispatch Design under Forecast Uncertainty (10-Generator System)}

We next introduce forecast uncertainty to examine how ramp adequacy screening can guide product and dispatch design. Forecast errors are Gaussian, with the 60-minute-ahead mean absolute error set to 5\% of mean demand and error dispersion increasing linearly with the forecast horizon. The 90\% forecast interval defines the upward uncertainty allowance $\delta(\tau)$. We conduct 30 Monte Carlo trials using common random numbers across all product configurations, with energy and reserve jointly re-cleared at every 5-minute interval.

\emph{Ramp Reserve Product Design.}
Forecast uncertainty changes both the magnitude and timing of the anticipated ramp requirement, so a product selected from a deterministic insufficient-duration set may not cover all realizations. The 30-minute product results in 2.11~MWh of mean shedding with a 2.27-MWh standard deviation, while the 60-minute product reduces these values to 0.57~MWh and 1.00~MWh, respectively. Both the ${30,60}$-minute and ${10,30,60}$-minute portfolios eliminate shedding in all 30 trials. These results show that when the insufficient-duration set shifts across forecast realizations, a portfolio spanning several relevant durations can provide more robust coverage than a single-duration product.

\begin{figure}[H]
\centering
\includegraphics[width=0.85\columnwidth]{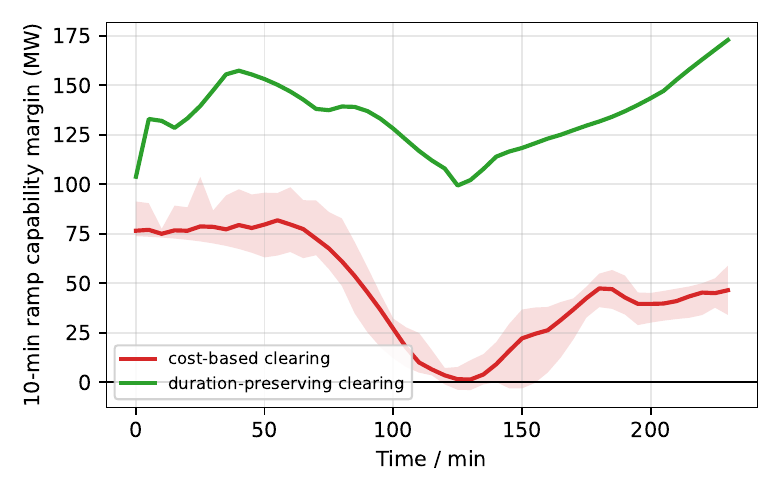}
\caption{Post-clearing dispatch-state sensitivity under the same 60-minute
product, 5\% forecast-error model.}
\label{fig:degradation}
\end{figure}

\emph{The Role of Dispatch Policy.} After ramp reserve procurements accounts for forecast uncertainty, dispatch positioning affects retained ramp capability. Figure~\ref{fig:degradation} compares two clearing policies under the same 60-minute product and forecast-error model. The cost-based case uses the standard production and load-shedding objective, but the duration-preserving case augments this objective with the remaining-duration dispersion penalty in \eqref{eq:sced_ramp_equalize_qp}, using $\beta=200$. Sensitivity tests at $\beta \in \{50, 100, 200, 400\}$ retain zero shedding in 10-generator cases and approximately 68\% shedding reduction in the Texas case with a production cost increase for $\beta=50\rightarrow 400$ of 0.7\% and 1.4\% for the 10-generator and Texas cases respectively.

The 10-minute temporal margin is used to evaluate the short-duration capability retained after clearing. Under cost-based clearing, mean shedding is 0.57~MWh, the minimum mean margin falls to 1.4~MW, and the 10th-percentile margin reaches $-3.9$~MW. Under duration-preserving clearing, mean shedding is effectively zero, while both the minimum mean and 10th-percentile margins remain at 99.4~MW.

A ramp-product requirement alone does not determine the capability retained after clearing. Ramp adequacy screening can therefore be used not only to select product durations, but also to evaluate whether the resulting dispatch policy preserves ramp capability at other relevant durations.

\subsection{Large-Scale Screening and Product Evaluation (Synthetic Texas Grid)}
\label{sec:texas}

The large-scale study uses the 2751-bus synthetic Texas system
\cite{birchfield2016grid}, with 877 generators, 5344 branches, and 96 15-minute intervals based on ERCOT's 2023 peak-load day. Demand is scaled by 1.135, solar capacity by 1.5, and transmission capacity by 1.5. Generator ramp rates are multiplied by factor 0.04 to create a ramp-constrained operating condition. Peak net demand is 77.6~GW, while conventional generation capacity is 77.4~GW. The dispatch model retains the full DC network constraints, which are equivalent to a PTDF formulation.  

\emph{Duration-Aware Ramp Adequacy Screening.}
Figure~\ref{fig:texas-screening}(a) fixes the baseline dispatch state at 13:00, the beginning of the selected afternoon ramp window. The requirement curve first exceeds the fleet capability curve at 360~minutes, identifying a multi-hour exposure and a 48.8-MWh lower bound on the corresponding security loss. This diagnosis indicates that the shortage is not primarily a short-duration ramp-rate problem. Instead, the committed fleet lacks sufficient sustained capability from its current dispatch position.

Figure~\ref{fig:texas-screening}(b) complements this diagnosis by tracking the 60-minute margin over time. Without a product, the margin first becomes negative at 17:45 and reaches a minimum of $-1713$~MW. The 60-minute product raises the minimum to $-1259$~MW, which improves the monitored margin, but its duration remains well below the 360-minute exposure identified by the fixed-time screen.

\emph{Implications for Product-duration Selection.}
The procurement rows of Table~\ref{tab:texas} confirm the screening result. The 30- and 60-minute products reduce shedding by 19\% and 18\%, respectively, while the $\{30,60\}$-minute portfolio reduces shedding by 28\%. These products provide partial value by improving dispatch positioning over shorter horizons, but none covers the full multi-hour exposure. The 120-minute product increases shedding by approximately 5\% because its repositioning is not sufficiently sustained relative to the diagnosed ramp event.

The case therefore illustrates an important product-design implication. Adding a longer product to an existing menu does not automatically improve security. Candidate durations should be selected according to the location of the insufficient-duration set and evaluated through their complete temporal-margin trajectories. When the diagnosed set extends beyond all available product durations, the operator may need a longer-duration portfolio, additional flexible resources, or a dispatch-side emergency response.

\begin{figure}[t]
\centering
\includegraphics[width=0.43\textwidth]{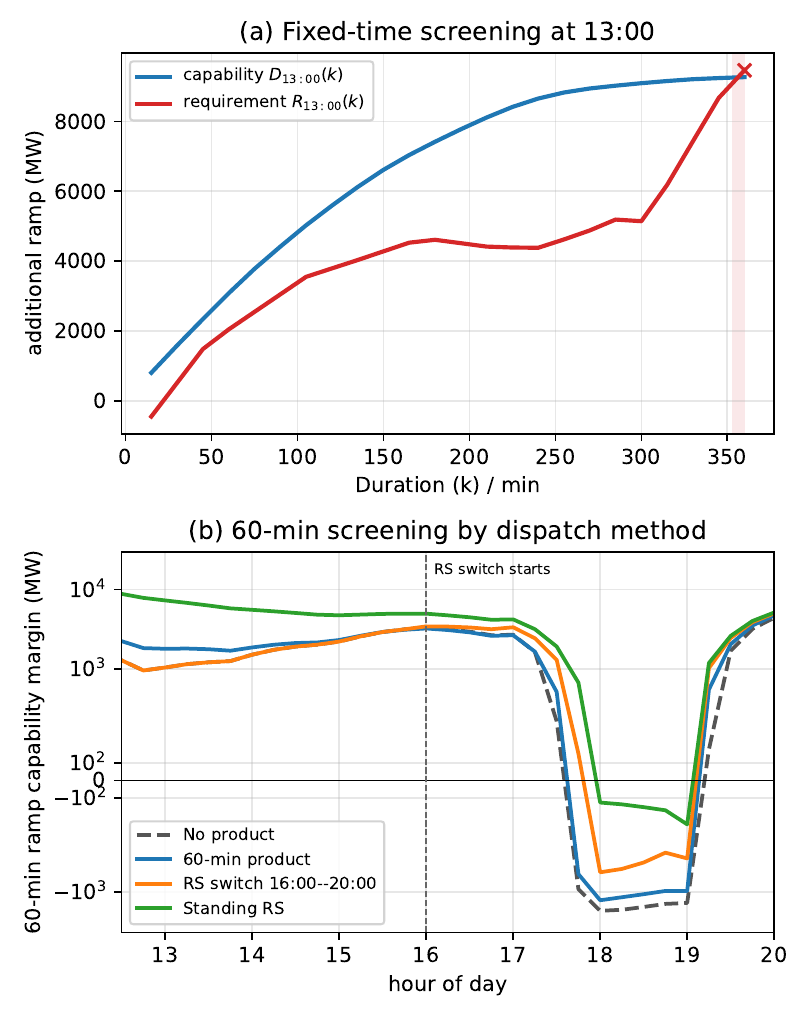}
\caption{Ramp adequacy screening on the Synthetic Texas grid.}
\label{fig:texas-screening}
\end{figure}

\subsection{Screening-Guided Emergency RS Dispatch (Synthetic Texas Grid)}

\emph{Emergency RS deployment.}
When ramp-product procurement does not fully address the diagnosed exposure, RS dispatch can serve as an emergency operating policy. Theorem~\ref{optimal} establishes RS dispatch as the minimum-security-loss policy in the transmission-unconstrained setting, motivating targeted activation rather than continuous use. For the synthetic Texas system, RS is implemented through the network-constrained RS-SCED formulation in \eqref{eq:sced_ramp_equalize_qp}, with $\beta=200$, so that remaining-duration preservation is optimized jointly with generator, nodal balance, and transmission constraints.

\emph{Screening-guided Emergency Activation.}
The fixed-duration temporal margin provides a direct signal for evaluating activation timing. In Fig.~\ref{fig:texas-screening}(b), switching to RS over $[16{:}00,20{:}00]$ before the baseline 60-minute margin becomes negative delays the first negative crossing from 17:45 to 18:00 and improves the minimum margin from $-1713$ to $-559$~MW. In contrast, the later $[18{:}00,22{:}00]$ activation begins after ramp scarcity has developed and provides substantially less protection.

Table~\ref{tab:texas} shows the corresponding security and cost outcomes. The $[16{:}00,20{:}00]$ activation reduces shedding by 68\%, from 2275 to 723~MWh, with a 2.9\% increase in production cost. The earlier and longer $[14{:}00,22{:}00]$ activation further reduces shedding to 471~MWh, whereas the late $[18{:}00,22{:}00]$ activation leaves 1929~MWh of shedding. Standing RS reduces shedding further to 322~MWh, but at the cost of continuously prioritizing ramp preservation.

\begin{table}[H]
\caption{Screening-Guided Ramp-Product Procurement and Emergency RS Dispatch
on the Synthetic Texas Grid}
\centering
\label{tab:texas}
\begin{threeparttable}
\setlength{\tabcolsep}{4pt}
\begin{tabular}{l|ccc}
\toprule
\textbf{Configuration}
& \textbf{Shed (MWh)}
& \textbf{Prod.\ (\$M)}
& \textbf{Total\tnote{*} (\$M)}
\\
\midrule
\multicolumn{4}{l}{\textit{Ramp-product procurement}}\\
No product, no switch       & 2275 & \textbf{10.76} & 33.51 \\
Product ${30}$ min        & 1838 & 10.79 & 29.17 \\
Product ${60}$ min        & 1858 & 10.82 & 29.40 \\
Portfolio ${30,60}$ min   & \textbf{1633} & 10.82 & \textbf{27.15} \\
Product ${120}$ min       & 2397 & 10.97 & 34.94 \\
\midrule
\multicolumn{4}{l}{\textit{Emergency RS dispatch}}\\
Switch $[14{:}00,18{:}00]$  & 756  & 11.28 & 18.84 \\
Switch $[16{:}00,20{:}00]$  & \textbf{723} & 11.08 & \textbf{18.30} \\
Switch $[18{:}00,22{:}00]$  & 1929 & \textbf{11.00} & 30.29 \\
\bottomrule
\end{tabular}
\begin{tablenotes}
\item[*] Total emergency cost $=$ production cost $+$ VOLL $\times$ shed load.
\end{tablenotes}
\end{threeparttable}
\end{table}

The experiments use predetermined activation windows rather than an automated margin threshold, but they demonstrate the value of screening-guided timing. Activating RS before the temporal margin becomes deeply negative preserves substantially more future ramp capability than attempting to restore the fleet after scarcity has developed. Residual shedding under RS-SCED reflects transmission congestion and system capacity limitations that are not fully captured by the fleet-level ramp adequacy margin.

\section{Discussion: Operation, Market, and Regulatory Implications}
\label{sec:disc}

Duration-aware ramp adequacy screening provides a state-dependent basis for ramp-product design and operation. The insufficient-duration band $\mathcal{B}_t$ identifies which ramp durations require additional protection under the current commitment, dispatch, and anticipated net-demand conditions. Historical and scenario-based distributions of $\mathcal{B}_t$ can therefore inform the standard product menu, while operational screening can identify which durations are relevant for a particular operating condition. When the exposed band spans multiple durations or varies across forecast scenarios, a portfolio may provide more robust protection than a single ramp product.

Screening also provides a post-clearing diagnostic because ramp capability evolves with the energy dispatch. Tracking the temporal margin $M_t(\bar{k})$ reveals whether the intended capability remains available after procurement and can provide an operational signal for emergency RS activation when product procurement is insufficient. Product certification should therefore account for the current commitment and dispatch state, including ramp rate, available headroom, and remaining ramp-up duration, rather than relying only on nameplate capability.

These applications raise several market-design considerations. Existing ramp-product settlements can compensate cleared capability and lost-opportunity costs \cite{chen2026ramping,looney2026lost}, while emergency RS activation may require transparent triggering and release criteria and ex post reporting. In addition, the current screening metric measures aggregate fleet capability and does not explicitly represent transmission congestion. Extending the framework to zonal or nodal screening and developing corresponding multi-duration procurement mechanisms remain important directions for future work.

\section{Conclusion}

This paper develops a duration-aware ramp adequacy screening framework based on the margin $M_t(k)=D_t(k)-R_t(k)$, which compares fleet ramp capability with the anticipated ramp requirement across duration. A negative margin identifies an insufficient-duration band in which the committed fleet lacks adequate ramp capability, providing a direct basis for ramp-product duration selection. Tracking the same margin over time further reveals how ramp adequacy evolves with dispatch and provides a metric for post-procurement evaluation.

The remaining-duration representation further reveals how ramp products improve security through dispatch positioning. Building on this insight, we develop ramp-reserve scarcity dispatch and show that, in the transmission-unconstrained setting, it achieves the same minimum operational security loss as the perfect-foresight benchmark. The case studies demonstrate that ramp adequacy screening can identify emerging ramp scarcity in advance, distinguish the value of different product durations and portfolios, quantify dispatch-state effects after procurement, and guide the timing of emergency RS activation.

Future work will extend the framework to zonal and nodal screening, examine ramp product demand curve design and explore pricing mechanisms for multi-duration ramp procurement and emergency dispatch.

\section*{Acknowledgment}
{The authors are grateful for many discussions with Congcong Wang, Long Zhao and Bin Huang from MISO. We also benefited from helpful comments and
critiques given by Kathleen Spees, Andrew Levitt, Sarah Sofia, and Natalie Northrup from the Brattle Group.}

\section*{Disclaimer}
{The views expressed in this paper are the opinion of the
authors and do not reflect the views of PJM Interconnection,
L.L.C. or its Board of Managers, of which Le Xie is a member.}

\appendices

\section{Greedy Algorithm and Proof of Theorem~\ref{optimal}}
\label{app:proof}

Using the remaining ramp-up duration in Definition~\ref{defin2}, the oracle problem \eqref{oracle} can be equivalently expressed in the remaining-duration
domain. Let
\begin{equation}
\label{eq:netramp}
\delta'(t)=\delta(t)+s(t-1),
\end{equation}
where $\delta(t)=d(t)-d(t-1)$ under Assumption~\ref{error}. The equivalent
oracle formulation $\mathcal{G}^{\mathrm{oracle}}_{\mathrm{eq}}$ is
\begin{equation}
\label{oracle2}
\begin{aligned}
& \min_{\{l_i(t)\},\,\{s(t)\}}
\quad  \sum_{t\in\mathcal T}s(t)\\
\text{s.t.}\quad
&0\leq l_i(t)\leq\overline l_i, \quad -1\leq l_i(t-1)-l_i(t)\leq1,\\
&\delta'(t)-s(t)
=\sum_{i=1}^{N}\overline r_i
\big[l_i(t-1)-l_i(t)\big],\\
&s(t)\geq0,
\qquad \forall i\in\mathcal N,\;t\in\mathcal T.
\end{aligned}
\end{equation}

The paper focuses on upward ramp scarcity and adopts the following condition.

\begin{assumption}[Non-binding Ramp-Down Capability]
\label{nocurt}
The system has enough downward flexibility over the intervals considered such that the ramp down constraints in the RS dispatch do not bind.
\end{assumption}

This assumption isolates the upward ramp-scarcity problem considered in Theorem~\ref{optimal}. The greedy algorithm operates directly on the remaining-duration state. Algorithm~\ref{greedyalg} allocates the required system ramp according to the RS priority rule described in Section~\ref{sec:backstop}, and its implementation requires running the algorithm $\forall t \in \mathcal{T}$.

\begin{algorithm}[h]
  \caption{Greedy Algorithm for Ramp-Reserve Scarcity Dispatch $\mathcal{G}^{\mathrm R}$ } \label{greedyalg}
  \SetAlgoLined
  \KwIn{Net ramping signal $\delta'(t)$; durations $\{l_i(t-1)\}_{i \in \mathcal{N}}$; parameters $\{\overline{r}_i, \overline{l}_i\}_{i \in \mathcal{N}}$}

\textbf{Candidate duration levels:}

  \For{\( i \in \mathcal{N} \)}{
    \( l_i'(t-1) = \max \{l_i(t-1) - 1, 0 \} \)  (ramp up)\\
    \( l_i''(t-1) = \min \{l_i(t-1) + 1, \overline{l}_i \} \) (ramp down)
  }

\textbf{Ordering:} Sort the initial and candidate remaining durations in descending order: \\
\( \{L_1(t),\dots,L_{3N}(t)\}=\text{sort}_\text{desc} \{l_{i}(t-1),l_i'(t-1),l_i''(t-1)\}_{i \in \mathcal{N}} \)

\textbf{Target duration:}
Find $[L_{j}(t),L_{j-1}(t)]$ whose total output changes bracket $\delta'(t)$ and interpolate to obtain $\underline L(t)$. If upward capability is insufficient, $\underline L(t)=\min_i l_i'(t-1)$ and shortfall is $s(t)$.

  \For{\( i \in \mathcal{N} \)}{
  \eIf{\( l_i(t-1) \geq \underline{L}(t) \)}{
    \( u_i(t) = \min \{ l_i(t-1) - \underline{L}(t), 1 \} \)
  }{
    \( u_i(t) = -\min \{ \underline{L}(t) - l_i(t-1), 1,  \overline{l}_i - l_i(t-1)\} \)
  }}
\textbf{State update:}
$l_i(t) = l_i(t-1)-u_i(t)$,
$s(t)   = \delta'(t) -\sum_{i \in \mathcal{N}} \overline{r}_i (l_i(t-1) - l_i(t))$

\KwOut{\( \{l_i(t), u_i(t)\}_{i \in \mathcal{N}}, s(t) \)}
\end{algorithm}

We exhibit, for the primal solution $\{l_i^*(t),u_i^*(t),s^*(t)\}$ returned by Algorithm~\ref{greedyalg}, a feasible dual solution of the oracle dispatch~\eqref{oracle} that satisfies complementary slackness; strong duality of the linear program then yields optimality.

\emph{Primal feasibility.} The update rules satisfy the power-capacity~\eqref{pc} and ramp~\eqref{rc} constraints. Under Assumption~\ref{nocurt}, the target-selection rule meets demand whenever feasible and assigns any shortfall to $s^*(t)$, ensuring power balance~\eqref{softpower} and $s^*(t)\ge0$. Hence the greedy solution is feasible for $\mathcal{G}^{\mathrm{oracle}}$.

\emph{Dual problem.} We seek a feasible dual point with $\rho_{it}^-=0$ (consistent with Assumption~\ref{nocurt}). For such multipliers, the dual objective is $\sum_t \lambda_t d(t) - \sum_{i,t}(\mu_{it}^+\overline g_i + \rho_{it}^+\overline r_i)-\sum_i\rho_{i1}^+g_i(0)$, with per-generator stationarity $-\mu_{it}^-+\mu_{it}^++\rho_{it}^+-\rho_{i,t+1}^+=\lambda_t$, the shedding condition $\lambda_t+\gamma_t=1$, nonnegativity of $\mu_{it}^\pm,\rho_{it}^\pm,\gamma_t$, and the boundary convention $\rho_{i,T+1}^\pm=0$. The balance multipliers $\lambda_t$ are unrestricted in sign.

\emph{Dual construction.} Let $\Xi^*(t)=\underline L^*(t)-\min_i\{l_i(t-1),l_i'(t-1),l_i''(t-1)\}$ measure operational slack. A value $\Xi^*(t)=0$ means that every generator is at its ramp-up or capacity limit, whereas $\Xi^*(t)>0$ implies $\underline L^*(t)>0$ and $s^*(t)=0$. Constructing the duals backward from $t=T$ gives
\begin{equation*}
\begin{aligned}
\lambda_t&=\begin{cases}1,&\Xi^*(t)=0,\\ -\min\limits_{i:g_i^*(t)>0}\rho_{i,t+1}^+, & \Xi^*(t)>0,\end{cases}\\[2pt]
\rho_{it}^+&=\begin{cases}\max\{0,\lambda_t+\rho_{i,t+1}^+\},&u_i^*(t)=1,\\ 0,&\text{otherwise,}\end{cases}\\[2pt]
\end{aligned}
\end{equation*}
For any $g_i^*(t)=0$, $\mu_{it}^-=\max\{0, -\lambda_t - \rho_{i,t+1}^+\}$ and is 0 otherwise. If all outputs $g^*_i(t)$ are zero and $\Xi^*(t)>0$, use $\lambda_t=-\max_i\rho_{i,t+1}^+$ instead of the empty minimum. Set $\mu_{it}^+=1+\rho_{i,t+1}^+$ when $g_i^*(t)=\overline g_i$, $u_i^*(t)<1$, and $\Xi^*(t)=0$ (and $0$ otherwise).

\emph{Feasibility and complementary slackness.} Under Assumption~\ref{nocurt}, the greedy update is
\[
l_i^*(t)=\min\{\overline l_i,\max\{l_i'(t-1),\underline L^*(t)\}\}.
\]
We prove by backward induction, that $l_i^*(t-1)\le l_j^*(t-1)$ implies $\rho_{it}^+\le\rho_{jt}^+.$. This holds at $t=T+1$ because the terminal multipliers are zero. Suppose it holds at $t+1$ and $l_i^*(t-1)\le l_j^*(t-1)$. If $\rho_{it}^+=0$, the conclusion follows from nonnegativity. Otherwise $u_i^*(t)=1$, so
\[
l_i'(t-1)=l_i^*(t-1)-1\ge\underline L^*(t).
\]
Consequently $l_j'(t-1)\ge l_i'(t-1)\ge\underline L^*(t)$, and the greedy update gives $u_j^*(t)=1$ and $l_i^*(t)\le l_j^*(t)$. The inductive hypothesis gives $\rho_{i,t+1}^+\le\rho_{j,t+1}^+$; the ramp-multiplier formula with the same $\lambda_t$ therefore gives $\rho_{it}^+\le\rho_{jt}^+$.

When $\Xi^*(t)>0$, positive-output generators have $l_i^*(t)\ge\underline L^*(t)$, while zero-output generators have $l_i^*(t)\le\underline L^*(t)$. The ordering and the definition of $\lambda_t$ consequently give $\lambda_t+\rho_{i,t+1}^+\ge0 \ \text{if }g_i^*(t)>0$ and $\lambda_t+\rho_{i,t+1}^+\le0 \ \text{if }g_i^*(t)=0.$ A positive-output generator with $u_i^*(t)<1$ has $l_i^*(t)=\underline L^*(t)$, so its next-period multiplier attains the minimum and $\lambda_t+\rho_{i,t+1}^+=0$. Thus a positive difference is assigned to a binding upward ramp limit and a negative difference to a binding lower output bound. This verifies stationarity.

When $\Xi^*(t)=0$, every generator is at its ramp-up or capacity limit. If $u_i^*(t)=1$, then $\rho_{it}^+=1+\rho_{i,t+1}^+$. Otherwise $g_i^*(t)=\overline g_i$, which implies $\rho_{i,t+1}^+=0$ because the next upward ramp constraint is nonbinding for a generator already at capacity; hence $\mu_{it}^+=1$. Stationarity follows in both cases.

All inequality multipliers are nonnegative, and $\lambda_t\le1$ gives $\gamma_t\ge0$. Every nonzero bound multiplier is assigned to a binding constraint. Finally, $\Xi^*(t)=0$ gives $\gamma_t=0$, whereas $\Xi^*(t)>0$ gives $s^*(t)=0$, so $\gamma_t s^*(t)=0$. The primal and dual solutions are feasible and satisfy complementary slackness. Therefore, Algorithm~\ref{greedyalg} attains the minimum operational security loss. \hfill$\blacksquare$

\bibliographystyle{IEEEtran}
\bibliography{ref.bib}

\end{document}